\documentclass[12pt,draftcls,onecolumn]{IEEEtran}
\usepackage{amssymb}
\usepackage{verbatim}
\usepackage{cite}
\ifCLASSINFOpdf
  \usepackage[pdftex]{graphicx}
\else
   \usepackage[dvips]{graphicx}
\fi
\usepackage[cmex10]{amsmath}
\usepackage{algorithmic}
\usepackage{array}
\usepackage{mdwmath}
\usepackage{mdwtab}
\usepackage{eqparbox}
\usepackage[tight,footnotesize]{subfigure}

\usepackage{fixltx2e}
\usepackage{stfloats}
\ifCLASSOPTIONcaptionsoff
  \usepackage[nomarkers]{endfloat}
 \let\MYoriglatexcaption\caption
 \renewcommand{\caption}[2][\relax]{\MYoriglatexcaption[#2]{#2}}
\fi
\usepackage{url}

\def\z{\mathbf{z}}
\def\M{\mathbf{M}}
\def\N{\mathbb{N}}
\def\N{\mathbb{N}}
\def\R{\mathbb{R}}
\def\Q{\mathbf{Q}}

\begin{document}
\newcommand{\tr}{\mbox{\tiny T}}
\newtheorem{property}{\noindent \textbf{Property}} {\normalfont }{\normalfont}
\newtheorem{result}{\noindent \textbf{Result}} {\normalfont }{\normalfont}
\newtheorem{alg2}{\noindent \textbf{Algorithm}} {\normalfont }{\normalfont}
\newtheorem{subalg}{\normalfont \textbf{A}}[alg2]{\normalfont}{\normalfont}
\newtheorem{subprop}{\normalfont \textbf{P}}[property]{\normalfont}{\normalfont}
\newenvironment*{proof2}[1]{\textbf{\emph{Proof}} }{}
\newtheorem{theorem}{\noindent \textbf{Theorem}}{\normalfont}{\normalfont}
\newtheorem{assumption}{Assumption}{\normalfont}{\normalfont}
\newtheorem{corollary}[theorem]{Corollary}{\normalfont}{\normalfont}
\newtheorem{lemma}{Lemma}{\normalfont}{\normalfont}
\newtheorem{prop}{\noindent \textbf{Proposition}}{\normalfont}{\normalfont}
\newtheorem{remark}{\noindent \textbf{Remark}}{\normalfont}{\normalfont}
\newtheorem{definition}{Definition}{\normalfont}{\normalfont}

%
\title{A unified framework for solving \\ a general class of conditional and robust \\ set-membership estimation problems$^{\diamond}$}
\author{Vito~Cerone$^{*}$,~\IEEEmembership{Member,~IEEE,}
        Jean-Bernard~Lasserre$^{\diamond}$,
        Dario~Piga$^{\S}$,~\IEEEmembership{Member,~IEEE,}
        Diego~Regruto$^{*}$,~\IEEEmembership{Member,~IEEE}
\thanks{$^{\diamond}$ Accepted for publication in IEEE Transactions on Automatic Control}
\thanks{$^{*}$ V. Cerone and D. Regruto are with Dipartimento di Automatica e Informatica, Politecnico di Torino, 10129 Torino, Italy
        {\tt\small vito.cerone@polito.it, diego.regruto@polito.it}}%
\thanks{$^{\diamond}$J. B. Lasserre is with LAAS-CNRS and Institute of Mathematics, University of Toulouse, 7 Avenue du Colonel Roche
31 077 Toulouse Cedex 4, France
        {\tt\small lasserre@laas.fr}}
\thanks{$^{\S}$D. Piga is with the Control Systems Group, Department of Electrical Engineering, Eindhoven
University of Technology, P.O. Box 513, 5600 MB, Eindhoven, The
Netherlands. {\tt\small D.Piga@tue.nl}}}
%

\maketitle

\vspace{-0cm}
\begin{abstract}

In this paper we present a unified framework for solving a general class of problems arising in the context of set-membership estimation/identification theory. More precisely, the paper aims at providing an original approach for the computation of optimal conditional and robust projection estimates in a nonlinear estimation setting where the operator relating the data and the parameter to be estimated is assumed to be a generic multivariate polynomial function and the uncertainties affecting the data are assumed to belong to semialgebraic sets. By noticing that the computation of both the conditional and the robust projection optimal estimators requires the solution to min-max optimization problems that share the same structure, we propose a unified two-stage approach based on semidefinite-relaxation techniques for solving such estimation problems. The key idea of the proposed procedure is to recognize that the optimal functional of the inner optimization problems can be approximated to any desired precision by a multivariate polynomial function by suitably exploiting recently proposed results in the field of parametric optimization. Two simulation examples are reported to show the effectiveness of the proposed approach.

\end{abstract}

\section{INTRODUCTION}
Estimation theory can roughly be defined as a branch of mathematics dealing with the problem of inferring the values of some unknown variables, usually called \textit{parameters}, from a set of empirical data related to the unknown parameters through a given, possibly uncertain, mathematical relation. Experimental data are usually obtained by means of measurement procedures that are known to be affected by uncertainty. Most of the results available in the estimation theory literature are based on a statistical description of the uncertainty affecting the data.

A worthwhile alternative to the stochastic description, is the so-called bounded-error or set-membership characterization where measurement errors are assumed to be \textit{unknown but bounded} (UBB), i.e., the measurement uncertainties are assumed to belong to a given bounded set. Such a description seems to be more suitable in those cases where either a priori statistical information is not available or the errors are better characterized in a deterministic way (e.g., systematic and class errors in measurement equipments, rounding and truncation errors in digital devices). Based on the UBB description of the uncertainty, a new paradigm called bounded-error or set-membership estimation has been proposed starting with the seminal work of Schweppe \cite{schw71}. In the last three decades, set-membership estimation theory has been the subject of extensive research efforts which led to a number of relevant results with emphasis on the application of the set-membership paradigm in the context of system identification. The interested reader is referred to the book \cite{minoplwa}, the survey papers \cite{milvic,wallah2} and the reference therein for a thorough review of the fundamental principles of the theory. Set-membership estimation algorithms can roughly be divided in two main categories: (i) set-valued estimators (see, e.g., \cite{milbel82,fohu82,pear88,ver94,viza96,cgvz98,cepire2011a,cepire2012a} and the references therein), aimed at deriving either exact or approximate descriptions of the so-called \textit{feasible parameter set}, i.e., the set of all possible parameter values consistent with the collected experimental data and a set of a-priori assumptions; (ii) pointwise estimators (see, e.g., \cite{MiTe85,kmvt1986,kamivi1988,mila95,ga1999,gaviza2000} and the references therein), that return a single element of the parameter space according to a given selection criteria.

In this paper we focus on the latter category and, in particular, on two classes of pointwise estimation algorithms called \textit{conditional estimators} and \textit{projection estimators} respectively. In a nutshell, a set-membership estimation algorithm is called a conditional estimator when the sought estimate is constrained to belong to a given set (see, e.g., \cite{kamivi1988,ga1999,gaviza2000,gkvz2000}), while it is called a projection estimator (see, e.g., \cite{kmvt1986,kamivi1988,mila95}) when the parameter estimate is sought by minimizing a certain norm of the so-called regression error. To the best of the authors' knowledge, most of the results presented in the bounded-error literature about conditional and/or projection estimation are derived under a set of simplifying hypotheses, including the assumptions that: (i) the operator relating the parameter and the experimental data is linear and is not affected by uncertainty, (ii) the error affecting the measured data belongs to simple-shaped convex sets (e.g. boxes, ellipsoids) and (iii) the parameter estimate to be computed is looked for in the entire parameter space or at most in a linearly parameterized subset of the parameter space.

In this work, by recognizing that the problems of computing the conditional and projection estimates require the solution to min-max optimization problems that share essentially the same structure, a unified approach is proposed to approximate to any desired precision the optimal (either conditional or projection) estimate by assuming that (i) the operator relating the parameter and the experimental data is a generic nonlinear polynomial function possibly dependent on a set of uncertain variables assumed to belong to a given semialgebraic set, (ii) the error affecting the measured data belongs to a semialgebraic set and (iii) the parameter estimate to be computed is sought in a semialgebraic subset of the parameter space. It is worth noticing that in full generality, solving nonconvex min-max optimization problems is a real challenge for which no general methodology is available.  An exception is a certain class of robust versions of some convex optimization problems when the uncertainty set has some special form. In this case, computationally tractable robust counterparts of these convex problems may exist. See for instance \cite{bental1,bental2} and the references therein.

The paper is organized as follows. The addressed estimation problem is formulated in Section II, where the proposed unified framework is also presented. A two-stage approach based on semidefinite-relaxation techniques for the solution of the considered class of estimation problems is then presented in Section III. The effectiveness of the proposed approach is demonstrated by means of two simulation examples in Section IV. Concluding remarks end the paper.

\section{PROBLEM FORMULATION}
In this paper we consider a class of parametric nonlinear set-membership estimation problems where a given nonlinear operator $\mathcal{F}$ maps the parameter $\theta \in \mathcal{P}_\theta \subseteq \mathbb{R}^\ell$ to be estimated into the output vector $\mathbf{w} \in \mathbb{R}^N$ as follows
\begin{equation}\label{model}
 \mathbf{w}= \mathcal{F}(\theta,\varepsilon_F)
\end{equation}
where $\varepsilon_F$ is an uncertain variable.
The set $\mathcal{P}_\theta$ takes into account possible prior information on the parameter $\theta$ to be estimated. In this work $\mathcal{P}_\theta$ is assumed to be a semialgebraic set of the form
\begin{equation}\label{theta_prior}
\mathcal{P}_\theta = \left\{\theta \in \mathbb{R}^\ell: k_z(\theta) \geq 0, z= 1,\ldots,s \right\}
\end{equation}
where $k_z,\ z=1,\ldots,m$ are multivariate polynomials in the $\ell$ components of the vector $\theta$.
Output measurements $\mathbf{y} \in \mathbb{R}^N$ are assumed to be corrupted by bounded noise as follows
\begin{equation}\label{output}
\mathbf{y} = g(\mathbf{w},\varepsilon_y)
\end{equation}
where $g$ is a polynomial function in the variable $\mathbf{w}$ and $\varepsilon_y$.
The uncertain variables $\varepsilon_F$ and $\varepsilon_y$ are assumed to belong to the following semialgebraic set
\begin{equation}\label{uncer}
\mathcal{S}_{\varepsilon_F,\varepsilon_y} = \left\{\varepsilon_F \in \mathbb{R}^q, \varepsilon_y \in \mathbb{R}^N: h_i(\varepsilon_F,\varepsilon_y) \geq 0, i = 1,\ldots,r \right\}
\end{equation}
with $h_i,\ i=1,\ldots,r$ being multivariate polynomials in the $q$ components of the vector $\varepsilon_F$ and the $N$ components of the vector $\varepsilon_y$.
In this work we restrict our attention to the case where the nonlinear operator $\mathcal{F}$ is a multivariate polynomial function of variables $\theta$ and $\varepsilon_F$. \\
In the set-membership estimation framework, all the values of $\theta$ that are consistent with the assumed model structure described in \eqref{model}, collected measurements $\mathbf{y}$ \eqref{output} and bounds on the uncertainty variables \eqref{uncer} 
are considered as feasible solutions to the estimation problem. The set $\mathcal{D}_\theta$ of all such values is called the \textit{feasible parameter set} (FPS) and can be defined as the projection onto the parameter space $\mathbb{R}^\ell$ of the following set $\mathcal{D}$:
\begin{equation}\label{augfps}
\begin{split}
\mathcal{D} = &\left\{(\theta,\varepsilon_F,\varepsilon_y) \in \mathcal{P}_\theta \times \mathbb{R}^q \times \mathbb{R}^{N}:\right.\\
&\left. \mathbf{y} = g(\mathcal{F}(\theta,\varepsilon_F),\varepsilon_y) ,\ (\varepsilon_F, \varepsilon_y) \in \mathcal{S}_{\varepsilon_F,\varepsilon_y}
\right\}.
\end{split}
\end{equation}
In this paper we provide a unified approach to address some minmax estimation problems arising in set-membership identification. In particular we will refer throughout the paper to the general formulation of the considered identification problems presented in Section \ref{gen_sec}.

\subsection{General min-max formulation of the considered class of SM estimation problems}\label{gen_sec}

The contribution of the paper is to provide an approach to solve the following general nonlinear set-membership estimation problems:
%
\begin{equation}\label{gen_minmax}
\textbf{P1:}\ \theta_{\mbox{\tiny rob}} = \arg \min_{\theta \in \mathcal{M}} \max_{\alpha \in \mathcal{S}_{\alpha,\theta}} J(\theta,\alpha)
\end{equation}
where $\mathcal{M}$ can be either $\mathcal{D}_\theta$ or $\mathcal{P}_\theta$ or any other possible subset of $\mathbb{R}^\ell$ described by a set of polynomial inequalities, $\alpha \in \mathbb{R}^T$,
%
\begin{equation}
\mathcal{S}_{\alpha,\theta} = \left\{\alpha \in \mathbb{R}^T: d_\mu(\alpha,\theta) \geq 0, \mu= 1,\ldots,M \right\}
\end{equation}
is a semialgebraic set, and $d_\mu,\ \mu=1,\ldots,M$ are multivariate polynomials in the components of the vectors $\alpha$ and $\theta$.
In the rest of the paper we will refer to \eqref{gen_minmax} as \textit{robust SM estimation problem} \textbf{P1}.
To the best of the authors' knowledge this is the first attempt towards the solution of a robust SM estimation problem in such a general form.

It is worth noting that computation of the global optimal solution $\theta_{\mbox{\tiny rob}}$ of problem \eqref{gen_minmax} is a difficult and challenging problem since \eqref{gen_minmax} is an NP-hard robust nonconvex optimization problem. As already mentioned, in full generality there is no methodology to solve \eqref{gen_minmax} except for
robust versions of some convex optimization problems
when the uncertainty set has some special form. Indeed, such problems have computationally tractable robust counterparts
as described, for example, in \cite{bental1,bental2}.

In sections \ref{cond_sec} and \ref{proj_sec} reported below, we show that the two classes of set-membership estimation problems considered in the paper (\textit{conditional central estimation} and \textit{robust projection estimation}) can be interpreted as two specific instances of problem \textbf{P1}.

As will be discussed in details in the following, the approach proposed in this paper to solve problem \textbf{P1} relies on the results presented in \cite{2010Ala} (see Theorem 1 of this paper) that were derived under the following assumption:
\\
\begin{assumption}\label{ass_empty}
For each fixed value of $\theta = \overline{\theta}$ the set $\mathcal{S}_{\alpha,\overline{\theta}}$ is nonempty. \\
\end{assumption}

In sections \ref{cond_sec} and \ref{proj_sec} we show that Assumption \ref{ass_empty} is always satisfied for the specific classes of estimation problems considered in the paper.
\subsection{Conditional central estimation for set-membership nonlinear errors-in-variables parametric identification}\label{cond_sec}
Consider a single-input single-output (SISO) nonlinear dynamic system  
which transforms the noise-free input  $x_t$ into the noise-free
output $w_t$ according to
%
\begin{equation}\label{eiv1}
w_t = f(\theta,w_{t-1},w_{t-2},\ldots w_{t-na}, x_{t-1},x_{t-2},\ldots x_{t-nb})
\end{equation}
where $\theta$ is the parameter vector to be estimated and $f$ is assumed to be a multivariate polynomial function. Both input and output data sequences are corrupted by additive noise, $\xi_t$ and $\eta_t$ respectively, i.e.
\begin{equation}\label{eiv2}
u_t  =  x_t + \xi_t, \; \; \; y_t  =  w_t + \eta_t.
\end{equation}
The noise samples $\xi_t$ and $\eta_t$ are bounded by given $\Delta\xi_t$ and $\Delta\eta_t$ respectively, that is:
\begin{equation}\label{eiv3}
\mid \xi_t \mid  \leq  \Delta \xi_t, \; \; \;  \mid \eta_t \mid  \leq  \Delta \eta_t.
\end{equation}
The nonlinear errors-in-variables (NEIV) model structure described by \eqref{eiv1}--\eqref{eiv3} can be written in the form \eqref{model} by setting:
\begin{equation}
{\theta} = \left[a_1 \;\;\; \ldots \;\;\; a_{na} \;\;\; b_0
\;\;\; b_1 \;\;\; \ldots \;\;\; b_{nb} \right]^{\tr},
\end{equation}
\begin{equation}
\mathbf{w} = [w_1\ w_2 \ldots w_N],
\end{equation}
\begin{equation}
\mathbf{y} = [y_1\ y_2 \ldots y_N],
\end{equation}
\begin{equation}\label{nlin_eiv_epsF}
\varepsilon_F = \left[\varepsilon_{F,1}\ \varepsilon_{F,2}\ \ldots\ \varepsilon_{F,N}\right]^{\tr}
\end{equation}
\begin{equation}
\varepsilon_{y,t} = [\eta_{1}\ \eta_{t-2} \ldots \eta_{N}],
\end{equation}
\begin{equation}\label{nlin_eiv_F}
\mathcal{F}(\theta,\varepsilon_F) = \left[\mathcal{F}_1(\theta,\varepsilon_{F,1})\ \mathcal{F}_2(\theta,\varepsilon_{F,2})\ \ldots\ \mathcal{F}_N(\theta,\varepsilon_{F,N})\right]^{\tr}
\end{equation}
where, for all $t=1,2,\ldots,N$,
\begin{equation}
\begin{split}
\mathcal{F}_t(\theta,\varepsilon_F) = f(\theta,y_{t-1}-\eta_{t-1},\ldots y_{t-na}-\eta_{t-na}, u_{t-1}-\xi_{t-1},\ldots u_{t-nb}-\xi_{t-nb}).
\end{split}
\end{equation}
and
\begin{equation}
\varepsilon_{F,t} = [\eta_{t-1}\ \eta_{t-2} \ldots \eta_{t-na}\ \xi_{t}\ \xi_{t-1} \ldots \xi_{t-nb}].
\end{equation}
The NEIV model structure considered in \eqref{eiv1} is quite general and comprises many important nonlinear model classes usually considered in system identification including, among the other, Hammerstein, Wiener and Lur'e models (see, e.g., \cite{cepire2012b,cepire2012c,cere2006}),  linear-parameter-varying (LPV) models with polynomial dependence on the scheduling variables \cite{toth_book,cepire_acc11}, polynomial nonlinear autoregressive (NARX) and nonlinear output error (NOE) models.
\\
In the NEIV bounded error problem considered here, the set $\mathcal{S}_{\varepsilon_F,\varepsilon_y}$ 
is simply described by the linear inequalities in \eqref{eiv3}, while $\mathcal{D}$ in \eqref{augfps} is the set of all parameter values and noise samples consistent with the collected experimental data, the model structure in \eqref{eiv1} and \eqref{eiv2}, and noise bounds in \eqref{eiv3}.
\\
As far as the class of linear time-invariant dynamic systems is considered, the nonlinear operator $\mathcal{F}(\theta,\varepsilon_F)$ simplifies to:
\begin{equation}\label{bilinF}
\mathcal{F}(\theta,\varepsilon_F) = F(\varepsilon_F)^{\tr}\theta
\end{equation}
where
\begin{equation}
F(\varepsilon_F) = \left[F_1\ F_2\ \ldots\ F_N\right]^{\tr}
\end{equation}
and, for all $t=1,2,\ldots,N$,
\begin{equation}
\begin{split}
F_t = &\left[-y_{t-1}+\eta_{t-1}\ -y_{t-2}+\eta_{t-2} \ldots  -y_{t-na}+\eta_{t-na} \right.\\
&\left.u_{t}-\xi_{t}\ \ u_{t-1}-\xi_{t-1} \ldots u_{t-nb}-\xi_{t-nb} \right].
\end{split}
\end{equation}
\\

As is well known, the linear EIV identification set-up (see \cite{2007Aso} for details) is quite general in the sense that many other common linear identification problems can be written in this framework. In fact, the problem of identifying an output error (OE) model is obtained by setting $\xi_t=0$, the case of finite-impulse-response (FIR) models is obtained for $n_a=0$, while the structure in \eqref{eiv1} and \eqref{eiv2} turns out to be an equation error (EE) model when $\xi_t=0$ and $\eta_t = \sum_{i=1}^{na}a_i\eta_{t-i}$.
\\

It is worth noting that in the general NEIV problems $\mathcal{D}$ is a nonconvex semialgebraic set since the constraints $\mathbf{y} - \mathcal{F}(\theta,\varepsilon_F) = \varepsilon_y$ in \eqref{augfps} are polynomial functions of $\theta$ and $\varepsilon_F$ and, moreover, the same property holds true in the simplified linear-time-invariant case, where constraints $\mathbf{y} - \mathcal{F}(\theta,\varepsilon_F) = \varepsilon_y$ in \eqref{augfps} are bilinear in $\theta$ and $\varepsilon_F$ due to \eqref{bilinF}.
\\

Although $\mathcal{D}$ is the set of all nonlinear dynamic models with structure \eqref{eiv1} that are consistent with experimental data and measurement error bounds, neither the feasible parameter set nor the tight outerbounding box derived in \cite{cepire2012a,cepire2011a} can be straightforwardly exploited for controller design or system behavior simulation. Thus, in many applications, the problem of selecting a single model among the feasible
ones arises. 
One of the most common choices in the SM literature is to look for the value 
of the parameter $\theta$ that minimizes the worst case $\ell_p$ estimation error computed over the entire feasible set, i.e. 
\begin{equation}\label{cest_eiv1}
\theta_c \doteq \arg \min_{\theta \in \mathbb{R}^\ell} \max_{(\theta_\nu,\varepsilon_F,\varepsilon_y) \in \mathcal{D_\nu}} \|\theta_\nu - \theta\|_p.
\end{equation}
where
\begin{equation}\label{augfps_nu}
\begin{split}
\mathcal{D_\nu} = &\left\{(\theta_\nu,\varepsilon_F,\varepsilon_y) \in \mathcal{P}_{\theta_\nu} \times \mathbb{R}^q \times \mathbb{R}^{N}:\right.\\
&\left. \mathbf{y} = g(\mathcal{F}(\theta_\nu,\varepsilon_F),\varepsilon_y) ,\ (\varepsilon_F, \varepsilon_y) \in \mathcal{S}_{\varepsilon_F,\varepsilon_y}
\right\},
\end{split}
\end{equation}
and $\|\cdot\|_p$ is the $\ell_p$-norm of a vector.

The estimate $\theta_c$ computed by solving \eqref{cest_eiv1} is the so-called $\ell_p$-\textit{Chebyshev center} of $\mathcal{D}$, also called \textit{central estimate} in the SM literature.
\\
\begin{remark}
In the case $p=\infty$, the central estimate is the center of the minimum-volume-box outerbounding $\mathcal{D}_\theta$ and can be computed by exploiting the convex relaxation approach proposed in \cite{cepire2012a}.
\\
\end{remark}

Although the central estimate provides the minimum of the worst-case estimation error, it may show some undesirable features in the case of EIV identification or, more generally, when the set $\mathcal{D}$ is nonconvex. More precisely, in those cases, the Chebyshev center $\theta_c$ is neither guaranteed to belong to the set $\mathcal{D}_\theta$ nor to the set $\mathcal{P}_\theta$ and, as a consequence, the identified LTI system could result inconsistent either with the experimental data or with some of the a-priori physical information on the parameter $\theta$.
In order to avoid such drawbacks, it is most desirable to force the computed parameter estimate to belong to a given set $\mathcal{M}$ by modifying the optimization problem \eqref{cest_eiv1} as follows
\begin{equation}\label{ccest}
\theta_c^{\mathcal{M}} \doteq \arg \min_{\theta \in \mathcal{M}} \max_{(\theta_\nu,\varepsilon_F,\varepsilon_y) \in \mathcal{D_\nu}} \|\theta_\nu - \theta\|_p
\end{equation}
where $\mathcal{M}\subset\mathbb{R}^\ell$ is assumed to be a semialgebraic set described by polynomial inequalities. Such a set is: (a) $\mathcal{D}_\theta$ if our aim is to constrain the computed estimate to belong to the feasible parameter set, (b) $\mathcal{P}_\theta$ if it is required to guarantee that the identified system satisfies the set of available a-priori information, (c) $\mathcal{D}_\theta \cap \mathcal{P}_\theta$, or (d) any other semialgebraic set if, more generally, we want to force the identified system to belong to a particular model class. Problem \eqref{ccest} falls into the class of \textit{Conditional set-membership estimation problems} and, in particular, $\theta_c^\mathcal{M}$ is referred to as the \textit{conditional Chebyshev center} of the feasible parameter set $\mathcal{D}$ with respect to the model class $\mathcal{M}$.
The problem of conditional central estimation is still a challenging problem in the field of set-membership identification/information based-complexity and a number of papers have appeared in the literature
in the last decades on the subject (we refer the reader to the paper \cite{gaviza2000} and
the references therein for a thorough review). In particular, conditional central algorithms have been proposed to effectively address the problems of reduced order modeling \cite{kamivi1988}, set-membership state smoothing and filtering \cite{gaviza1999} and worst-case identification \cite{gaviza2000}. For such problems, computationally efficient and/or closed-form solutions to the problem of conditional central estimation have been derived assuming that: (i) $\mathcal{F}$ is a linear operator in both $\theta$ and $\varepsilon_F$, (ii) $\varepsilon_F=0$, (iii) $\mathcal{M}$ is a linear manifold, and (iv) $\mathcal{S}_{\varepsilon_y}$ is a simple-shaped convex set (usually a box, an ellipsoid or a polytope). Unfortunately, such assumptions are not satisfied in many relevant identification problems including, for example, the EIV problem considered in this section. As an additional motivating example leading to the class of estimation problems defined in \eqref{ccest}, we mention the problem of identifying input-output linear systems that are a-priori known to be bounded-input bounded-output (BIBO) stable. In this case, we are interested in computing the optimal estimate of the system parameter, in the Chebyshev center sense, over the set $\mathcal{P}_{stab}$ of all the parameter values that guarantee BIBO stability of the system. Since, as shown in \cite{cepire2011b}, the set $\mathcal{P}_{stab}$ is semialgebraic and described by polynomial inequalities, such a problem naturally leads to a conditional estimation problem of the general form \eqref{ccest} where the set $\mathcal{M} = \mathcal{P}_{stab}$.
\\
\begin{remark}
It is worth remarking that in problem \eqref{ccest} $\mathcal{S}_{\alpha,\theta}$ coincides with $\mathcal{D_\nu}$ where $\alpha = (\theta_\nu,\varepsilon_F,\varepsilon_y)$. Therefore, since $\mathcal{D_\nu}$ does not depend on $\theta$, Assumption \ref{ass_empty} is satisfied as long as $\mathcal{D_\nu}$ is a nonempty set, that is a common assumption in Set-membership identification, often satisfied in practice unless the identification problem is not well posed (e.g. the considered a-priori assumption on the system to be identified are completely wrong).
\end{remark}

\subsection{Robust conditional projection estimation}\label{proj_sec}
Another class of estimator of particular interest in the set-membership/information-based complexity (IBC) framework is given by the so-called projection algorithms (see, e.g., \cite{kmvt1986,mila95} and the references therein) where the parameter estimate is computed by solving an optimization problem of the following form:
\begin{equation}\label{proj}
\theta_p = \arg \min_{\theta \in \mathcal{M}}\|e_R(\varepsilon_F,\theta)\|_p
\end{equation}
where $e_R = \mathbf{y}-\mathcal{F}(\theta,\varepsilon_F)$ is the so-called regression error (see, e.g., \cite{1999Alj,1989Asost}) and the set $\mathcal{M}$ either coincides with $\mathbb{R}^\ell$ or is a subset of $\mathbb{R}^\ell$. In the latter case the obtained estimator is called a \textit{conditional projection algorithm} (see, e.g., \cite{kamivi1988}). The optimization criterion $\|e_R\|_p$ is widely adopted in the identification literature (see, e.g., \cite{1999Alj,1983Ny,AHL1996} and the references therein) and, in particular, the popular least-square estimator (see, e.g., \cite{1999Alj,1989Asost}) is obtained by setting $p=2$, $\mathcal{M} = \mathbb{R}^\ell$ and assuming that $\varepsilon_F = \mathbf{0}$ where $\mathbf{0}$ is the null element of the space $\mathbb{R}^\ell$. Projection estimators and their optimality properties have been extensively investigated in the SM framework and a number of interesting results have been derived (see, e.g., \cite{kmvt1986,mila95,gkvz2000}). However, to the best of the authors' knowledge, most of such results have been obtained under the assumptions that: (i) $\mathcal{F}$ is a linear operator in both $\theta$ and $\varepsilon_F$, (ii) $\mathcal{M}$ is a linear manifold, (iii) $\mathcal{S}_{\varepsilon_y}$ is a simple-shaped convex set (usually a box, an ellipsoid or a polytope) and, most important, (iv) assuming that the operator $\mathcal{F}$ is not affected by uncertainty, i.e., $\varepsilon_F=0$. In this work, we consider the following generalization of \eqref{proj}
\begin{equation}\label{robproj}
\theta_p^r = \arg \min_{\theta \in \mathcal{M}} \max_{(\varepsilon_F,\varepsilon_y) \in \mathcal{S}_{\varepsilon_F,\varepsilon_y}}\|e_R(\theta,\varepsilon_F)\|_p
\end{equation}
where, in order to take care of the effects of the uncertainty affecting the problem, we look for the parameter estimate $\theta_p^r$ that minimizes the worst case $\ell_p$ regression error computed over the entire uncertainty set $\mathcal{S}_{\varepsilon_F,\varepsilon_y}$. As far as the set $\mathcal{M}$ is concerned, we only assume that $\mathcal{M}$ be a subset of $\mathbb{R}^\ell$ described by polynomial inequalities. In such a way, the user is allowed, for example, to constrain the optimal estimate to belong to the feasible parameter set ($\mathcal{M}=\mathcal{D}_\theta$), to guarantee that $\theta_p^r$ satisfies the available a-priori physical information ($\mathcal{M}=\mathcal{P}_\theta$) or, more generally, to force the estimated model to belong to a specific, possibly reduced-order, model class (see, e.g., \cite{kamivi1988,gkvz2000}). In the rest of the paper we refer to $\theta_p^r$ as the \textit{robust p-norm projection estimate} or RPE for short. It is worth noticing that the problem of computing the optimal projection estimate for the case $p=2$ (least squares estimate) in the presence of an uncertainty $\varepsilon_F \neq 0$ under the restricting assumptions that $\mathcal{F}$ is linear operator and the set $\mathcal{S}_{\varepsilon_F,\varepsilon_y}$ is convex, has been widely studied also outside the context of set-membership estimation and a number of different approaches have been proposed (see, e.g., paper \cite{elgleb97} and the references therein for a thorough review of the available methods and results).
\\
\begin{remark}
It is worth remarking that in problem \eqref{robproj} $\mathcal{S}_{\alpha,\theta}$ coincides with $\mathcal{S}_{\varepsilon_F,\varepsilon_y}$ where $\alpha = (\varepsilon_F,\varepsilon_y)$. Therefore, since $\mathcal{S}_{\varepsilon_F,\varepsilon_y}$ does not depend on $\theta$, Assumption \ref{ass_empty} is satisfied for all $\theta \in \mathbb{R}^\ell$ as long as $\mathcal{S}_{\varepsilon_F,\varepsilon_y}$ is a nonempty set, that is a common assumption in Set-membership identification, often satisfied in practice unless the considered a-priori assumption on the measurement errors are completely wrong.
\end{remark}

\section{A semidefinite relaxation approach}

In this section a two-stage approach is proposed to approximate to any desired precision the global optimal solution to the general SM robust identification problem \eqref{gen_minmax}. The proposed approach is based on the following basic observations:
\begin{itemize}
\item[(i)] problem \eqref{gen_minmax} is a \textit{two-players non-cooperative game} (see, e.g., \cite{Nash51}) where $\mathcal{M}$ and $\mathcal{S}_{\alpha,\theta}$ are the action sets of the first and the second player respectively;
\item[(ii)] from the point of view of the second player, \textbf{P1} is a \textit{parametric optimization problem} (see \cite{2010Ala} and the references therein) in the sense that the optimal value
of  the inner maximization problem in (\ref{gen_minmax}) is a function of the decision of player 1, i.e. the value of the parameter $\theta$;
\item[(iii)] the optimal value function $\widetilde{J}$ of the parametric inner maximization problem is given by
    \begin{equation}\label{maxproblem}
        \widetilde{J}(\theta) \doteq J(\theta,\alpha^*(\theta)) = \max_{\alpha \in \mathcal{S}_{\alpha,\theta}} J(\theta,\alpha)
    \end{equation}
    and it is a function of parameter $\theta$ only.
\end{itemize}
Thanks to observations (i)--(iii) above, once $\widetilde{J}(\theta)$ is known, problem \textbf{P1} simplifies to the following optimization problem:
\begin{equation}\label{gen_minmax_ver2}
    \textbf{P2:}\ \theta_{\mbox{\tiny rob}} = \arg \min_{\theta \in \mathcal{M}}\widetilde{J}(\theta).
\end{equation}
Unfortunately, a general methodology to derive an exact closed-form expression for  function $\widetilde{J}(\theta)$ is not available and, therefore, a two-stage procedure is proposed here to approximate the global optimal solution of problem \textbf{P1}. In the first stage, a polynomial function
$\widetilde{J}^*_\tau(\theta)$ of degree $2\tau$, upper approximating (in a strong sense) the optimal value function $\widetilde{J}(\theta)$ of the parametric inner maximization problem is computed. Then, in the second stage, problem \textbf{P2} is replaced with the following polynomial optimization problem:
\begin{equation}\label{gen_minmax_ver3}
    \textbf{P3:}\ \theta_{\tau} = \arg \min_{\theta \in \mathcal{M}}\widetilde{J}^*_\tau(\theta).
\end{equation}

\subsection{Polynomial approximation of the function $\widetilde{J}(\theta)$}\label{pol_approx}
Polynomial approximation of the function $\widetilde{J}(\theta)$ is performed here by exploiting the methodology for parametric polynomial optimization proposed in \cite{2010Ala}. However, in order to apply the results presented in \cite{2010Ala}, we first need to compute a set $\mathcal{R}_\theta \subset \mathbb{R}^\ell$ outerbounding $\mathcal{M}$, whose shape is simple enough to allow one to easily compute all the moments of a Borel probability measure $\varphi$ with uniform distribution on $\mathcal{R}_\theta$. In this work we exploit the SDP-relaxation based procedure proposed in \cite{cepire2012a} to compute the minimum-volume axis-aligned box containing the set $\mathcal{M}$. Once the box $\mathcal{R}_\theta=[\underline{\theta},\overline{\theta}]\subset\mathbb{R}^\ell$, i.e.,
\[\mathcal{R}_\theta\,=\,\{\theta\in\mathbb{R}^\ell\::\:\phi_k(\theta)\geq0,\: \phi_k(\theta)\doteq(\overline{\theta_k}-\theta_k)(\theta_k-\underline{\theta}_k),\quad k=1,\ldots,\ell\}\]
has been computed, we can 
formulate the following optimization problem where we look for the upper polynomial approximation $\widetilde{J}_\tau(\theta)$ of the optimal value function $\widetilde{J}(\theta)$ such that the integral $\int_{\mathcal{R}_\theta} \widetilde{J}_\tau(\theta)d\varphi(\theta)$ is minimized:
\begin{equation}
\begin{split}
&\min_{\tilde{J}_\tau(\theta)} \int_{\mathcal{R}_\theta} \tilde{J}_\tau(\theta)d\varphi
\\
&\mbox{s.t.}\ \tilde{J}_\tau(\theta) \geq J(\theta,\alpha)\ \ \forall (\theta,\alpha) \in \mathcal{R}_\theta \times \mathcal{S}_{\alpha,\theta}.
\end{split}
\end{equation}
By noticing that (i) the objective function can be written as a linear combination of the moments of the uniform distribution measure supported on $\mathcal{R}_\theta$ and (ii) the inequality constraint can be approximately replaced by a SOS constraint, the following semidefinite relaxed problem is obtained (\cite{2010Ala}):
\begin{equation}\label{SOS_relax}
\left\{
\begin{split}
\min_{\lambda_\beta,\sigma_\mu,\psi_k} &\sum_{\beta \in \mathbb{N}_{2\tau}^\ell} \lambda_{\beta}\gamma_\beta \\
\mbox{s.t.}\ &
\sum_{\beta \in \mathbb{N}_{2i}^p}\lambda_{\beta}\theta^\beta - J(\theta,\alpha) = \sigma_0(\theta,\alpha)+\sum_{\mu=1}^M \sigma_\mu(\theta,\alpha)d_\mu(\alpha,\theta)+\sum_{k=1}^\ell\psi_k(\theta,\alpha)\,\phi_k(\theta)\\
& \sigma_\mu \subset \Sigma[\theta,\alpha],\ \mu=1, \ldots, M \\
& \psi_k \subset \Sigma[\theta,\alpha],\ k=1, \ldots, \ell \\
& \deg(\sigma_0) \leq 2\tau;\ \deg(\sigma_\mu d_\mu) \leq 2\tau,\ \mu=1, \ldots, M,\ \deg(\psi_k \phi_k)\leq 2\tau,\ k=1, \ldots, \ell,
\end{split}
\right.
\end{equation}
where for each $\beta = [\beta_1 \ldots \beta_\ell] \in \mathbb{N}^\ell$ and $\theta = [\theta_1 \ldots \theta_\ell]$ the notation $\theta^\beta$ stands for the monomial $\theta_1^{\beta_1} \theta_2^{\beta_2} \ldots \theta_\ell^{\beta_\ell}$, $\mathbb{N}_{2\tau}^\ell = \{\beta \in \mathbb{N}^\ell:\sum_j \beta_j \leq 2\tau\}$, $\Sigma[\theta,\alpha]$ is the set of SOS polynomials in the variables $\theta$ and $\alpha$, while $\gamma_\beta$ are the moments of the Borel probability measure $\varphi$ with uniform distribution on $\mathcal{R}_\theta$, defined as (see, e.g., \cite{2010Lasbook}):
\begin{equation}
\gamma_\beta \doteq \int_{\mathcal{R}_\theta} \theta^{\beta}d\varphi(\theta).
\end{equation}
\begin{lemma}
\label{lemma-existence}
If $\mathcal{R}_\theta\times\mathcal{S}_{\alpha,\theta}$ contains an open set, then the semidefinite program (\ref{SOS_relax}) has an optimal solution
$(\overline{\lambda}_\beta^*,\overline{\sigma}_\mu^*,\overline{\psi}_k^*)$, $\mu=0, \ldots, M$, $k=1,\ldots,\ell$, provided that tau is sufficiently large.
\end{lemma}
A detailed proof of Lemma \ref{lemma-existence} is postponed to the Appendix.

Next, by applying the results presented in \cite{2010Ala} about parametric polynomial optimization, it is possible to show that the optimal solution of \eqref{SOS_relax} enjoys the important property stated in the following theorem:
\begin{theorem}\label{approxL1}(\!\!\cite{2010Ala})
Let $(\overline{\lambda}_\beta^*,\overline{\sigma}_\mu^*,\overline{\psi}_k^*)$, $\mu=0, \ldots, M$, $k=1,\ldots,\ell$, be an optimal solution of problem \eqref{SOS_relax} for a given degree $\tau$ and let us define the polynomial $\widetilde{J}^*_\tau(\theta) = \sum_{\beta \in \mathbb{N}_{2\tau}^\ell} \overline{\lambda}_{\beta}^*\theta^\beta$. Then, $\widetilde{J}^*_\tau(\theta)$ converges to the optimal value function $\widetilde{J}(\theta)$ for the $L_1(\mathcal{R}_\theta,\varphi)$-norm as $\tau$ goes to infinity, i.e.:
\begin{equation}
\int_{\mathcal{R}_\theta} |\widetilde{J}^*_\tau(\theta)-\widetilde{J}(\theta)\vert\,d\varphi(\theta) \rightarrow 0.
\end{equation}
\end{theorem}
For the proof of Theorem \ref{approxL1} we refer the reader to the paper \cite{2010Ala}. We also have the following
property.

\begin{prop}
\label{usc}
The optimal value function $\theta\mapsto \widetilde{J}(\theta)$ is upper semicontinuous (u.s.c.) on $\mathcal{R}_\theta$.
\end{prop}
\begin{proof}
Let $(\theta_n)\subset\mathcal{R}_\theta$ be a sequence such that $\theta_n\to\theta$ as $n\to\infty$, and
\[\limsup_{z\to\theta} \widetilde{J}(z)\,=\,\lim_{n\to\infty} \widetilde{J}(\theta_{n}).\]
Next, for each $n\in\mathbb{N}$, let $\alpha^*(\theta_{n})$ be an arbitrary maximizer in $\mathcal{S}_{\alpha,\theta}$ for the max problem in (\ref{maxproblem}).
By compactness of $\mathcal{R}_\theta$ and $\mathcal{S}_{\alpha,\theta}$, there is a subsequence denoted $(n_\ell)$, $\ell\in\mathbb{N}$,
and a point $(\alpha,\theta)\in\mathcal{S}_{\alpha,\theta}\times\mathcal{R}_\theta$ such that $(\alpha^*(\theta_{n_\ell}),\theta_{n_\ell})\to(\alpha,\theta)$ as $\ell\to\infty$.
Consequently, using continuity of $J$,
\[\limsup_{z\to\theta} \widetilde{J}(z)\,=\,\lim_{n\to\infty} \widetilde{J}(\theta_{n})\,=\,
\lim_{\ell\to\infty}J(\theta_{n_\ell},\alpha^*(\theta_{n_\ell}))\,=\,J(\theta,\alpha)\,\leq \,\max_{\alpha\in\mathcal{S}_{\alpha,\theta}}\,J(\theta,\alpha)\,=\,\widetilde{J}(\theta),\]
which proves that $\widetilde{J}$ is u.s.c.
\end{proof}
Thanks to Theorem \ref{approxL1} we are in the position of proving the following result, which shows that the solution to problem $\textbf{P3}$ converges to the solution of $\textbf{P2}$ (and hence problem $\textbf{P1}$) as $\tau$ goes to infinity.

\begin{theorem}\label{add1}
Let $\widetilde{J}^*_\tau(\theta)$, $\tau\in\mathbb{N}$, be the polynomial defined in Theorem \ref{approxL1}. Consider the polynomial optimization problem
$\textbf{P3}$ in (\ref{gen_minmax_ver3}) with optimal value denoted by $J^*_\tau$, and let $\theta^*_\tau\in\mathcal{M}$
be an optimal solution of $\textbf{P3}$.  Let $\widehat{J}_\tau\,=\,\min_{k\leq \tau}J^*_k\,=\,\widetilde{J}^*_{k(\tau)}(\theta^*_{k(\tau)})$ for some $k(\tau)\in[1,\ldots,\tau]$.

Then:
\begin{equation}
\label{result1}
\lim_{\tau\to\infty}\left( \min_{k\leq \tau}\:J^*_k\right)\,=\,
\lim_{\tau\to\infty}\widehat{J}_\tau\,=\,
\min_{\theta\in\mathcal{M}}\,\max_{\alpha\in \mathcal{S}_{\alpha,\theta}}\,J(\theta,\alpha)\,=:\,J^*.
\end{equation}
Moreover, if $\widetilde{J}(\theta)$ is continuous on $\mathcal{M}$ and
$\theta_{\mbox{\tiny rob}}$ in (\ref{gen_minmax_ver2}) is unique, then
$\theta^*_{k(\tau)}\to\theta_{\mbox{\tiny rob}}$ as $\tau\to\infty$. If $\theta_{\mbox{\tiny rob}}$
is not unique then any accumulation point of the sequence $(\theta^*_{k(\tau)})$, $\tau\in\N$, is a global minimizer
of problem $\min\{\tilde{J}(\theta):\theta\in\mathcal{M}\}$.
\end{theorem}
\begin{proof}
Observe that being $\widetilde{J}^*_\tau(\theta)$ continuous on $\mathcal{R}_\theta$ (hence on $\mathcal{M}$),
it  has a global minimizer $\theta^*_\tau\in\mathcal{M}$, for every $\tau$.
From Theorem \ref{approxL1}, $\widetilde{J}^*_\tau(\theta)\stackrel{L_1(\mathcal{R}_\theta,\varphi)}{\to} \widetilde{J}(\theta)$
(i.e., convergence in the $L_1(\mathcal{R}_\theta,\varphi)$-norm).
Hence, by \cite[Theorem 2.5.3]{ash}, there exists a subsequence
$(\tau_\ell)$ such that $\widetilde{J}^*_{\tau_\ell}(\theta)\to \widetilde{J}(\theta)$, $\varphi$-almost uniformly on $\mathcal{R}_\theta$.

Next, by Proposition \ref{usc}, the optimal value mapping $\widetilde{J}$ is u.s.c. on $\mathcal{R}_\theta$
(hence on $\mathcal{M}$).
With $\epsilon>0$ fixed, arbitrary, let $\mathbf{B}(\epsilon)\doteq\{\theta\in\mathcal{M}: \widetilde{J}(\theta)<J^*+\epsilon\}$ and
let $\kappa\doteq\varphi(\mathbf{B}(\epsilon))$. As $\widetilde{J}$ is u.s.c.,
$\mathbf{B}(\epsilon)$ is nonempty, open, and therefore $\kappa>0$.
As $\widetilde{J}^*_{\tau_\ell}(\theta)\to \widetilde{J}(\theta)$, $\varphi$-almost uniformly on $\mathcal{R}_\theta$, there exists a Borel set $A_\kappa\in\mathcal{B}(\mathcal{R}_\theta)$ such that
$\varphi(A_\kappa)<\kappa$ and $\widetilde{J}^*_{\tau_\ell}(\theta)\to \widetilde{J}(\theta)$, uniformly on $\mathcal{R}_\theta\setminus A_\kappa$.
Hence, as $\Delta\doteq(\mathcal{R}_\theta\setminus A_\kappa)\cap \mathbf{B}(\epsilon)\neq\emptyset$, one has
\[\lim_{\ell\to\infty}\widetilde{J}^*_{\tau_\ell}(\theta)\,=\, \widetilde{J}(\theta)\,\leq\,J^*+\epsilon,\qquad\forall \,\theta\in\Delta,\]
and so, as $J^*_\tau\leq \widetilde{J}^*_\tau(\theta)$ on $\Delta$, one obtains $\displaystyle\lim_{\ell\to\infty}J^*_{\tau_\ell}
\leq J^*+\epsilon$. As $\epsilon>0$ was arbitrary, one finally gets
 $\displaystyle\lim_{\ell\to\infty}J^*_{\tau_\ell}=J^*$.
 On the other hand, by monotonicity of the sequence $(\widehat{J}_\tau)$,
\begin{equation}
\label{auxiliary1}
J^*\leq \displaystyle\lim_{\tau\to\infty}\widehat{J}_\tau\,=\,\displaystyle\lim_{\ell\to\infty}
\widehat{J}_{\tau_\ell}\,\leq\,\displaystyle\lim_{\ell\to\infty}J^*_{\tau_\ell}\,=\,J^*,
\end{equation}
and so (\ref{result1}) holds.

Next, let $\theta^*_\tau\in\mathcal{M}$ be a global minimizer of $\widetilde{J}^*_\tau(\theta)$ on $\mathcal{M}$.
As $\mathcal{M}$ is compact, there exists $\overline{\theta}\in\mathcal{M}$ and a subsequence $\tau_\ell$
such that $\theta^*_{k(\tau_\ell)}\to\overline{\theta}$ as $\ell\to\infty$. In addition,
from $\widetilde{J}^*_\tau(\theta)\geq \widetilde{J}(\theta)$ for every $\tau$ and every $\theta\in\mathcal{R}_\theta$,
\[J^*\,\leq\,\widetilde{J}(\theta^*_{k(\tau_\ell)})\,\leq\, \widetilde{J}^*_{k(\tau_\ell)}(\theta^*_{k(\tau_\ell)})\,=\,
\widehat{J}_{\tau_\ell}.\]
So using (\ref{auxiliary1}) and letting $\ell\to\infty$ yields the desired result
$J^*=\widetilde{J}(\overline{\theta})$. So if the minimizer of $\widetilde{J}$ on $\mathcal{M}$ is unique,
one has $\overline{\theta}=\theta_{\mbox{\tiny rob}}$, and as the converging subsequence $(\tau_\ell)$ was
arbitrary, the desired result follows. If the minimizer is not unique then every accumulation point
$\overline{\theta}$ is a global minimizer since $J^*=\widetilde{J}(\overline{\theta})$, as just shown above.
\end{proof}
\begin{remark}
Even though the sequence $\tilde{J}^*_\tau\to \tilde{J}$ for the $L_1$-norm, the sequence $\tilde{J}^*_\tau$, $\tau\in\N$, is  not necessarily monotone
(meaning $\tilde{J}^*_\tau(\theta)\geq\tilde{J}^*_{\tau+1}(\theta)$ for all $\tau\in\N$, $\theta\in\mathcal{R}_\theta$, does not necessarily holds). For this reason, it would be useful to know bounds on the distance between $\tilde{J}^*_\tau$ and $\tilde{J}$ for each fixed value of $\tau$. Unfortunately, computation of such bounds is a difficult open problem that requires further investigation.
\end{remark}
\subsection{Solution to problem \textbf{P3} via SDP relaxation}\label{pol_relax}
Once a polynomial approximation $\widetilde{J}^*_\tau(\theta)$ of the optimal value function $\widetilde{J}(\theta)$ of the inner maximization problem in \eqref{gen_minmax} has been computed as discussed in \ref{pol_approx}, we are in the position of solving problem \textbf{P3} which is a multivariate polynomial optimization problem in the variable $\theta$ on the compact semi-algebraic set $\mathcal{M}$. By applying the moments-based relaxation approach proposed in \cite{2001Ala}, a hierarchy of SDP relaxations $(\mathbf{Q}_t)$, $t\in\mathbb{N}$, can be constructed with the following properties:

- The resulting sequence $(\inf\mathbf{Q}_t)$, $t\in\mathbb{N}$, of optimal values is monotone non decreasing and converges to the optimal value $J^*_\tau$ of problem
$\textbf{P3}$.

- If the global minimizer $\theta_\tau\in\mathcal{M}$ of $\textbf{P3}$ is unique then the vector of ``first-order" moments of an optimal solution
$\mathbf{y}$ of $\mathbf{Q}_t$ converges to $\theta_\tau$ as $t\to\infty$.

For more details on SDP relaxations for generalized moment problems, the interested reader is referred, e.g., to \cite{2010Lasbook}.
In fact, in view of recent results in \cite{marshall} and \cite{nie}, the convergence is finite provided that the problem satisfies a set of mild conditions (see \cite{nie} and the references therein for details), that is, \textit{generically}, the optimal value $J^*_\tau$ is attained at a particular relaxation in the hierarchy, i.e., $J^*_\tau=\inf\mathbf{Q}_t$ for some $t$. Finally, another recent result
by Nie \cite{nie-cert} ensures that \textit{generically}, eventually some rank test is passed at some step $t$ in the hierarchy,
which permits to detect finite convergence at step $t$, and extract global minimizers (which \textit{generically} are finitely many). The reader is referred to the papers \cite{nie,nie-cert} for a discussion on the precise technical meaning of the word \textit{generically} in this context.
\\

\begin{remark}\label{exp:sp}\textbf{[Exploiting sparsity]}
At a first sight, the applicability of the relaxation-based procedure proposed in this paper seems to be limited in practice to small-size identification problems, due to large dimensions of the SDP problems involved in the two stages of the proposed approach. That is certainly true for problem \textbf{P1} in its general form \eqref{gen_minmax}, where the functional $J(\theta,\alpha)$ is a generic multivariate polynomial and the sets $\mathcal{M}$ and $\mathcal{S}_{\alpha,\theta}$ are generic semialgebraic sets. However, in view of the discussion and results reported in works \cite{cepire2012a,cepire2012b}, it is possible to show that a number of identification problems arising from real-word applications enjoy a peculiar sparsity structure, called \textit{correlative sparsity} in the framework of large-scale optimization (see, e.g., \cite{2006Awakikomu,2008kokiko}), which can be exploited to significantly reduce the computational complexity and the size of the involved SDP optimization problems either by means of the approach proposed in \cite{2006Awakikomu,2006Ala} or by means of the ad-hoc procedure presented in \cite{cepire2011a}. More specifically, it is possible to show that a number of set-membership identification problems leads to semialgebraic optimization where the constraints and the functional satisfy the so-called \textit{running intersection property} (see \cite{2006Ala}), a condition that guarantees convergence of the solution of the relaxed problem to the global optimum of the polynomial problem also when the correlative sparsity pattern is used to derive semidefinite relaxations of reduced complexity (see, e.g., \cite{2006Awakikomu,2008kokiko}). Analysis of the correlative sparsity structure of problem \textbf{P1} cannot be performed in general, since it requires to precisely specify the mathematical structure of the sets $\mathcal{M}$ and $\mathcal{S}_{\alpha,\theta}$. At the same time, providing a general discussion on the subject of sparsity exploitation in the context of SDP relaxation for polynomial problems is far beyond the scope of the paper, and the interested reader is referred to papers \cite{2006Awakikomu,2006Ala,2008kokiko}. However, we will try to provide here a sketch of the main ideas, covering the subject mostly at the level of intuition. \\ Let $\{1, \ldots n\}$  be the union $\bigcup_{k=1}^p I_k$ of subsets $I_k \subset \{1, \ldots n\}$. A polynomial optimization problem is said to enjoy a correlative type of sparsity structure if: (i) each polynomial involved in the description of the set of constraints is only concerned with variables $\{X_i: i \in I_k\}$ for some $k$; (ii) the functional to be optimized $J$ can be written as the sum $J=J_1+\ldots+J_p$ such that each $J_k$ only involves variables $\{X_i: i \in I_k\}$. Furthermore, the problem satisfies the running intersection property if the following condition is fulfilled:
$$
I_{k+1} \cap \bigcup_{j=1}^k I_j \subseteq I_s,\mbox{for some}\ s \leq k
$$
\end{remark}
The subsets $\{I_k\}$ can be detected either by inspection or by exploiting the systematic approach proposed in \cite{2006Awakikomu} and implemented in the software package \cite{2008Awakikimu}. If the problem enjoys a correlative sparsity structure, this can be used to derive SDP relaxations of lower complexity, as described in \cite{2006Awakikomu,2006Ala}. Essentially, the intuitive idea underlying the approaches proposed in \cite{2006Awakikomu,2006Ala}  is the following:  if the constraints and the objective function can be properly decomposed in subsets/subfunctionals depending only on a small subset of  variables, then ``sparse'' SDP relaxations can be constructed. This means that the involved SOS polynomials depend, each one, only on a small subset of variables  of the original polynomial optimization problem. The fact that the linear EIV identification and the nonlinear Hammerstein identification problems enjoy a correlative sparsity structure satisfying the running intersection property, has been proved in previous papers \cite{cepire2012a,cepire2012b}. The same arguments/reasoning can be used to show that the polynomial approximation/optimization problems obtained by applying the approach proposed in this paper to the problem of conditional central estimation problem \eqref{ccest}, enjoy the correlative sparsity structure and satisfy the running intersection property for many different choices of the set $\mathcal{M}$ including, e.g., the case $\mathcal{M}=\mathcal{D}$. This is also true for a number of problems in the class of robust conditional projection estimators including the nonlinear nonconvex robust least squares problem considered in Example 2 of Section IV.

\begin{remark}
It is worth-remarking that, by exploiting recent results presented in \cite{2010lapu}, the two-stage relaxation-based procedure proposed in this section can be extended to a more general class of problems where the function $J(\theta,\alpha)$ in problem $\textbf{P1}$ is a non-polynomial semialgebraic function.
\end{remark}


\section{Simulation examples}

The capabilities of the presented approach are shown in this section by means of two simulation examples.\\

\noindent \emph{Example 1}

The first illustrative example comes from the problem addressed in \cite{CaGaVi2011} on the identification of ARX models
based on quantized measurements. Consider the system analyzed in \cite{CaGaVi2011}, i.e.,
\begin{equation} \label{eqn:exsystem}
\mathbf{w}(t)=\theta_1^{\mathrm{o}}\mathbf{w}(t-1)+\theta_2^{\mathrm{o}}u(t)+d(t)=0.6\mathbf{w}(t-1)+0.6u(t)+d(t),
\end{equation}
where $u(t)$ and $\mathbf{w}(t)$ are the input and output signals at time $t$, respectively, and  $d(t)$ is an unknown additive disturbance which is assumed to belong to the interval $[-0.1, \ 0.1]$. The system is simulated using a white input signal $u(t)$  uniformly distributed within $[-2.5, \ 2.5]$ and a disturbance $d(t)$ with uniform distribution in the interval $[-0.1, \ 0.1]$. The output $\mathbf{w}(t)$ is measured by a  binary sensor with threshold $C=1$, i.e.,
\begin{equation}
\mathbf{y}(t)=\left\{\begin{array}{ccc}
             1 & \mbox{ if } \mathbf{w}(t) \geq 1 \\
             0 & \mbox{ otherwise }
           \end{array}
\right.
\end{equation}
where $\mathbf{y}(t)$ is the output of the binary sensor. Indeed, the system output $\mathbf{w}(t)$ is not accessible and only its measurement $\mathbf{y}(t)$ is available.
The estimate of the parameters $\theta$ of the system in  \eqref{eqn:exsystem} is computed based on a collection of $N=200$ input/output measurements.
 Note that $\mathbf{y}(t)$ can be written in the form of \eqref{output} as follows:
\begin{equation} \label{eqn:ywep}
\mathbf{y}(t)=\mathbf{w}(t)+\varepsilon_y(t),
\end{equation}
with $\varepsilon_y(t)$ s.t.
\begin{equation} \label{eqn:exnoise}
\begin{array}{ll}
  \varepsilon_y(t) \leq 0 & \mbox{ if } \mathbf{y}(t)=1, \\
  \varepsilon_y(t) \geq -1 & \mbox{ if } \mathbf{y}(t)=0.
\end{array}
\end{equation}
Based on eqs. \eqref{eqn:exnoise}, the uncertainty set $\mathcal{S}_{\varepsilon_y}$      can be written in terms of nonnegative inequality constraints as
\begin{equation}\label{exnoise}
\mathcal{S}_{\varepsilon_y} = \left\{\varepsilon_y \in \mathbb{R}^N: h_t(\varepsilon_y) \geq 0,\  t = 1,\ldots,N \right\},
\end{equation}
with
\begin{equation} \label{eqn:exnoiseconstraints}
 h_t(\varepsilon_y)\doteq  \left\{\begin{array}{ll}
 - \varepsilon_y(t) & \mbox{ if } y(t)=1, \\
  \varepsilon_y(t) + 1 & \mbox{ if } y(t)=0.
\end{array} \right.
\end{equation}
Substitution of eq.  \eqref{eqn:ywep} into \eqref{eqn:exsystem} leads to the following relation between input and noise-corrupted output $\mathbf{y}(t)$:
\begin{equation} \label{inpuoutrel}
\mathbf{y}(t)=\theta_1\left(\mathbf{y}(t-1)- \varepsilon_y(t-1) \right)+\theta_2u(t)+d(t)+\varepsilon_y(t).
\end{equation}
The FPS  $\mathcal{D}_\theta$ for the considered system is thus defined as the projection over the parameter space of the set $\mathcal{D}$ defined by \eqref{inpuoutrel}, \eqref{exnoise} and the a-priori assumption on the disturbance $d(t)$, i.e.
\begin{equation} \label{ex:DD}
\begin{split}
\mathcal{D} =  \left\{(\theta, d, \varepsilon_y) \in \mathbb{R}^{2+2N}:\ \right. & \mathbf{y}(t)=\theta_1\left(\mathbf{y}(t-1)- \varepsilon_y(t-1) \right)+\theta_2u(t)+d(t)+\varepsilon_y(t),  \\
 &  \left. h_t(\varepsilon_y) \geq 0, \ \ -0.1 \leq d(t) \leq 0.1, \; \; t=1,\ldots,N \right\}.
\end{split}
\end{equation}
Note that $\mathcal{D}$ is described by polynomial constraints because of the product between the unknown parameter $\theta_1$ and the noise $\varepsilon_y(t-1)$ in the equality constraint appearing in \eqref{ex:DD}. In this example we will compute the $\ell_2$-norm conditional Chebyshev center $\theta_c^{\mathcal{D}}$ of the FPS $\mathcal{D}$ with respect to  $\mathcal{D}_\theta$ itself, i.e.,
\begin{equation}\label{ex:ccest}
\theta_c^{\mathcal{D}} \doteq \arg \min_{\theta \in \mathcal{D}_\theta} \max_{(\theta_\nu,d,\varepsilon_y) \in \mathcal{D}} \|\theta_\nu - \theta\|_2^2.
\end{equation}
In order to compute a solution to problem \eqref{ex:ccest} through  the procedure discussed in the paper, an outer-bounding box $\mathcal{R}_\theta$ of the FPS $\mathcal{D}_\theta$ is first evaluated by means of the approach proposed in \cite{cepire2012a} for bounding the parameters of linear systems in the bounded-error EIV framework. The computed outer-bounding box $\mathcal{R}_\theta$ is reported in Fig. \ref{F:FPSCC}, together with the true FPS $\mathcal{D}_\theta$. Then, a polynomial  $\widetilde{J}_\tau(\theta)^*$ of degree $2\tau$ (with $\tau=2$) upper approximating the function
\begin{equation} \label{eqn:truepo}
\widetilde{J}(\theta)= \max_{(\theta_\nu,d,\varepsilon_y) \in \mathcal{D}} \|\theta_\nu - \theta\|_2^2,
\end{equation}
is computed by solving the SDP problem \eqref{SOS_relax}. It is worth remarking that problem \eqref{SOS_relax} enjoys a particular structured sparsity which is used to reduce the computational complexity in constructing the SOS polynomials in  \eqref{SOS_relax}. In fact, the objective function $\|\theta_\nu - \theta\|_2^2$ in \eqref{eqn:truepo} only depends on the model parameters $\theta_\nu$, while each constraint defining $\mathcal{D}$ in \eqref{ex:DD} only depends on a small subset of variables, namely, the model parameters $\theta_\nu$, the disturbance $d(t)$ and the noise samples $\varepsilon_y(t-1)$ and $\varepsilon_y(t)$.  A correlative sparsity structure satisfying  conditions in Remark \ref{exp:sp} can be easily detected through a procedure similar to the one discussed in \cite{cepire2012a} in the context of set-membership EIV identification.
\begin{figure}[!h]
\centerline{
\includegraphics[scale=1]{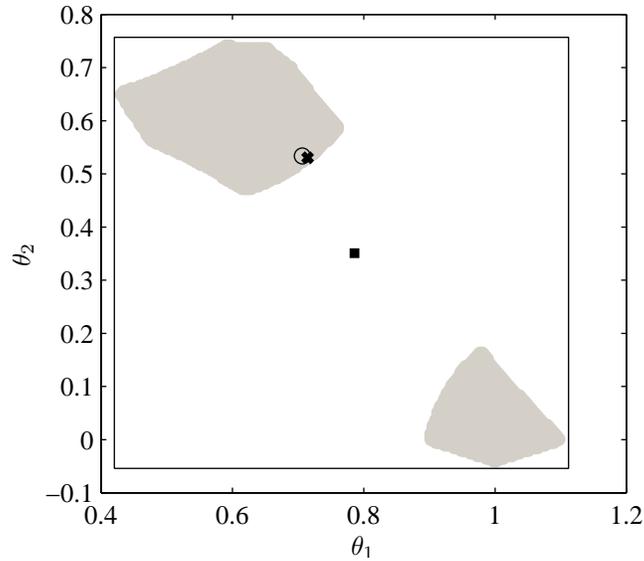}
}\caption{Exact Feasible Parameter Set  $\mathcal{D}_\theta$ (grey region), outer-bounding box $\mathcal{R}_\theta$ (region inside the box), (unconditional) Chebyshev center ({\tiny  $\blacksquare$}), exact conditional Chebyshev center $\theta_c^{\mathcal{D}}$ (\textbf{$\times$}), approximation of the conditional Chebyshev center computed with the proposed two-stage approach (O).}\label{F:FPSCC}
\end{figure}
The obtained 4-degree polynomial $\widetilde{J}_\tau(\theta)^*$ given by
\begin{align}
\widetilde{J}_\tau(\theta)^*= & 0.939-0.795\theta_1-0.037\theta_2-1.039\theta_1^2-3.731\theta_1\theta_2+ 4.315\theta_2^2+ \nonumber\\
                             & 2.617\theta_1^3+0.790\theta_1^2\theta_2+ 2.473\theta_1\theta_2^2-4.567\theta_2^3 + \\
                              & -0.968\theta_1^4+1.740\theta_1^3\theta_2-7.267\theta_1^2\theta_2^2+ 7.917\theta_1\theta_2^3-1.166\theta_2^4, \nonumber
\end{align}
is plotted in Fig. \ref{F:Pol}, together with the true function $\widetilde{J}(\theta)$ in \eqref{eqn:truepo}, which in turn has been obtained by gridding.
\begin{figure}[h!]
\centerline{
\includegraphics[scale=1]{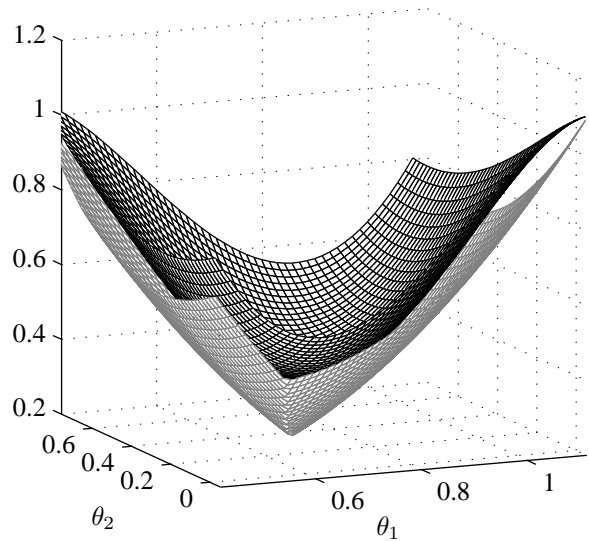}
}\caption{True function $\widetilde{J}(\theta)$ in \eqref{eqn:truepo} (gray) and computed polynomial approximation $\widetilde{J}^*_\tau(\theta)$ (black).}\label{F:Pol}
\end{figure}

The (unconditional) Chebyshev center of the FPS is computed by minimizing $\widetilde{J}_\tau(\theta)^*$ over the outer-bounding box $\mathcal{R}_\theta$, while an approximation of the conditional Chebyshev center is computed by solving problem \textbf{P3} via the SDP relaxation approach discussed in Section \ref{pol_relax}. Both the unconditional and the conditional Chebyshev center are reported in Fig. \ref{F:FPSCC}, which shows that the unconditional one does not belong to the FPS, while the computed approximation of the conditional Chebyshev center does. In the same figure,  the exact conditional Chebyshev center $\theta_c^{\mathcal{D}_\theta}$, that is, the minimum of the true function $\widetilde{J}(\theta)$ over the exact FPS  $\mathcal{D}_\theta$ is also reported showing that the proposed relaxation approach is able to provide a good approximation of the global optimal solution to problem \eqref{ex:ccest}. The CPU time taken to compute the conditional Chebyshev center  $\theta_c^{\mathcal{D}_\theta}$ is about 1320 seconds on a 2.40-GHz Intel Pentium
IV with 3 GB of RAM. More specifically, the time required to compute the polynomial approximation $\widetilde{J}(\theta)$ (i.e. time required to compute the solution to problem $\eqref{SOS_relax}$) is about $450$ seconds, while the second step (solution to minimization problem \textbf{P3} with order of relaxation $2$) takes about $870$ seconds. The maximum amount of memory used by Matlab during the computation was about 891 MB. The solver SeDuMi has been used to solve the SDP problem \eqref{SOS_relax} and the SDP problems relaxing \textbf{P3}.\\

\noindent \emph{Example 2}

In this example, the method is applied to the problem of robust estimation of a non-linear-in-the-parameter static model when both the input and the output measurements are corrupted by bounded noise.

 The  \emph{multi-input-single-output} (MISO) data-generating system is  given by
\begin{align}
\mathbf{w}(t)& =\theta_1^{\mathrm{o}}x_1(t)+\theta_1^{\mathrm{o}}\theta_2^{\mathrm{o}}x_2(t)+\theta_3^{\mathrm{o}}x_3(t)+\left(\theta_1^{\mathrm{o}}\right)^2x_4(t)+\theta_4^{\mathrm{o}}\theta_5^{\mathrm{o}}x_5(t)+\left(\theta_5^{\mathrm{o}}\right)^2x_6(t)+\theta_4^{\mathrm{o}}\theta_6^{\mathrm{o}}x_7(t)=\\
             & = 1x_1(t)+(1\cdot0.6)x_2(t)-0.5x_3(t)+1x_4(t)+0.3\cdot 0.8x_5(t)+\left(0.8\right)^2\!x_6(t)-0.3\cdot 0.5x_7(t),
\end{align}
where $x_i(t)$, with $i=1,\ldots,7$, is the $i$-th noise-free input and $\mathbf{w}(t)$ is the noise-free output at time $t$.
 The inputs $x_i(t)$ are i.i.d. random processes uniformly distributed in the interval $[-1, \ 1]$ with length $N=400$.
  Both the inputs  $x_i(t)$ and  the output $\mathbf{w}(t)$ are corrupted by additive uncertainties $\xi_i(t)$ and $\eta(t)$, respectively, i.e.,
\begin{align}
u_i(t)=& x_i(t)+\xi_i(t), \; \; \; \; \; \; i=1,\ldots,7, \\
\mathbf{y}(t)=& \textbf{w}(t)+\eta(t),
\end{align}
where $\xi_i(t)$ and $\eta(t)$ are white-noise processes uniformly distributed in the intervals $[-\Delta \xi_i, \ \Delta \xi_i]=[-0.2, \ 0.2]$ (for all $i=1,\ldots,7$) and $[-\Delta \eta, \ \Delta \eta]=[-0.25, \ 0.25]$, respectively.  The signal-to-noise ratio on the inputs, $\mathrm{SNR}_{x_i}$, and on output, $\mathrm{SNR}_{\mathbf{w}}$, defined as
\begin{equation}
\mathrm{SNR}_{x_i}=10\log \left\{  \sum_{t=1}^{N}x_i^{2}(t) \right/ \left.
                       \sum_{t=1}^{N}\xi_i^{2}(t) \right\},
\end{equation}
\begin{equation}
\mathrm{SNR}_{\mathbf{w}}=10\log \left\{  \sum_{t=1}^{N}\mathbf{w}^{2}(t) \right/ \left.
                       \sum_{t=1}^{N}\eta^{2}(t) \right\},
\end{equation}
are 13 db (for all $i=1,\ldots,7$) and 16 db, respectively. Let us denote with $\theta = [\theta_1, \ \theta_2, \ \theta_3, \ \theta_4, \ \theta_5, \ \theta_6]$ the parameters of the model to be estimated. The FPS $\mathcal{D}_\theta$ is then given by the projection over the parameter space of the following set:
\begin{equation} \label{ex:DDex2}
\begin{split}
\mathcal{D} =  \left\{(\theta, \xi, \eta) \in \mathbb{R}^{3+5N}:\ \right. & \mathbf{y}(t)=\theta_1\left(x_1(t)-\xi_1(t)\right)+\theta_1\theta_2\left(x_2(t)-\xi_2(t)\right)+\\
& +\theta_3\left(x_3(t)-\xi_3(t)\right)+\theta_1^2\left(x_4(t)-\xi_4(t)\right)+ \\
& +\theta_4^{\mathrm{}}\theta_5^{\mathrm{}}\left(x_5(t)-\xi_5(t)\right)+\left(\theta_5^{\mathrm{}}\right)^2\left(x_6(t)-\xi_6(t)\right)+\theta_4^{\mathrm{}}\theta_6^{\mathrm{}}\left(x_7(t)-\xi_7(t)\right)+\\
& + \eta(t),  \\
 &  \left. |\eta(t)|\leq \Delta\eta, \ \ |\xi_i(t)|\leq \Delta\xi_i, \; \; t=1,\ldots,N, \ \ i=1,\ldots,7 \right\}.
\end{split}
\end{equation}

Now, let $\hat{\mathbf{y}}(t,\theta)$ be the output of the model to be estimated, given by:
  \begin{align}
\hat{\mathbf{y}}(t,\theta)& =\theta_1^{}u_1(t)+\theta_1\theta_2u_2(t)+\theta_3u_3(t)+\theta_1^2u_4(t)+\theta_4^{\mathrm{}}\theta_5^{\mathrm{}}u_5(t)+\left(\theta_5^{\mathrm{}}\right)^2u_6(t)+\theta_4^{\mathrm{}}\theta_6^{\mathrm{}}u_7(t).
\end{align}
In this example, we compute the parameter estimate $\theta^*=[\theta^*_1, \ \theta^*_2, \ \theta^*_3, \ \theta^*_4, \ \theta^*_5, \ \theta^*_6]$ that minimizes the worst-case $\ell_2$-loss function  $\mathcal{V}(\theta,\xi)$, defined as
\begin{equation} \label{eqn:ese2V}
\mathcal{V}(\theta,\xi)=\sum_{t=1}^N \left(\mathbf{y}(t)-\hat{\mathbf{y}}(t,\theta)\right)^2,
\end{equation}
over all possible realizations of the input uncertainties $\xi_i(t)$ in the interval $[-\Delta \xi_i, \ \Delta \xi_i]$ under the constraint that the identified parameters belong to the FPS. The considered estimation problem can be formulated as the following \emph{min-max} optimization problem:
  \begin{equation} \label{eqn:minmaxex2}
  \hat{\theta}^*=\mathrm{arg}\min_{\theta \in \mathcal{D}_{}}\max_{\xi \in \mathcal{S}_{\xi}}\mathcal{V}(\theta,\xi),
  \end{equation}
  where $\mathcal{S}_{\xi}$ is defined as
  \begin{equation} \label{Sxi}
  \begin{split}
  \mathcal{S}_{\xi}=\left\{\xi: |\xi_i(t)|\leq \Delta \xi_i, \ \ \mbox{for all \ } i=1,\ldots,7, \ \ t=1,\ldots,N\right\}.
  \end{split}
  \end{equation}

It is worth noting that problem \eqref{eqn:minmaxex2} is: (i) a nonlinear nonconvex least squares problems, due to the nonlinear-in-parameter structure of the system to be estimated; (ii) a robust nonlinear least-square problem, due to the presence of uncertainty in all the explanatory variables; (iii) a nonconvex constrained least square problem, since the optimal estimate is looked for over the feasible parameter set $\mathcal{D}_\theta$. Therefore, problem \eqref{eqn:minmaxex2} is a challenging estimation problem for which, to the best of the authors' knowledge, no solution has been previously proposed in the literature.

Here, the solution to Problem \eqref{eqn:minmaxex2} is computed by applying the two-stage relaxation based method presented in the paper, which leads to the following estimate of the model parameters:
\begin{equation}
\hat{\theta}^*=\left[\hat{\theta}_1^*,\ \hat{\theta}_2^*,\ \hat{\theta}_3^*, \ \hat{\theta}_4^*, \ \hat{\theta}_5^*, \ \hat{\theta}_6^*\right]=\left[0.98, \ 0.57, \ -0.54, \ 0.36, \ 0.79, \ -0.59\right].
\end{equation}
It is worth remarking that, in order to apply the proposed method, an outer-bounding box of the feasible set  $\mathcal{D}$ has been computed by suitable modifications of the algorithm proposed in \cite{cepire2011a}. Furthermore, as in Example 1,  problem \eqref{eqn:minmaxex2} enjoys a particular sparsity structure which has been  exploited to reduce the computational load in solving \eqref{eqn:minmaxex2}. In fact, the objective function $\mathcal{V}(\theta,\xi)$ is given by the sum of $N$ terms $\left(\mathbf{y}(t)-\hat{\mathbf{y}}(t,\theta)\right)^2$, each one involving only the model parameters $\theta$ and the noise samples $\xi_i(t)$ (with $i=1,\ldots,7$) as unknown variables. Furthermore, each constraint defining  $\mathcal{S}_{\xi}$ in \eqref{Sxi} only depends on the noise variable $\xi_i(t)$. Similarly, each constraint defining the set $\mathcal{D}$ in \eqref{ex:DDex2} only involves a small subset of optimization variables, namely: the model parameters $\theta$, the input noise samples $\xi_i(t)$ (with $i=1,\ldots,7$) and the output noise sample $\eta(t)$. By stacking the  variables $\theta,\xi_i(t),\eta(t)$ in the vector
\begin{equation}
X=\left[\theta_1, \ \ldots, \ \theta_6, \ \xi_1(1), \ldots, \ \xi_7(1), \ldots, \ \xi_1(N), \ldots, \ \xi_7(N), \ \eta(1), \ \ldots, \ \eta(N) \right].
\end{equation}
The index sets $I_t$ (with $t=1,\ldots,N$) introduced in Remark \ref{exp:sp} and satisfying the \emph{running intersection property} can be defined as
\begin{equation}
I_t=\left\{1,\ldots,6,6+7(t-1)+1,\ldots,6+7(t-1)+7,6+7N+t\right\}.
\end{equation}
In this way, each constraint defining $\mathcal{D}$ in \eqref{ex:DDex2} is only concerned with  variables $\left\{X_i: i \in I_t\right\}$, that is  $\theta$,  $\xi_i(t)$ (with $i=1,\ldots,7$) and  $\eta(t)$.

The performance of the estimated
 model is tested on a validation set with $N_{\mathrm{val}}=100$ input/output measurements. The noise-free output  $\mathbf{w}(t)$ and the estimated output signal $\hat{\mathbf{y}}(t,\hat{\theta}^*)$   are plotted in Fig. \ref{Fese2:output}, while  the difference between $\mathbf{w}(t)$ and  $\hat{\mathbf{y}}(t,\hat{\theta}^*)$ is depicted in Fig. \ref{Fese2:err} showing a good agreement between the two signals. The CPU time taken to compute the parameter estimate $\hat{\theta}^*$ is about $7$ hours. More specifically, the time required to compute the solution to problem $\eqref{SOS_relax}$ with $\tau=2$ is about $2.5$ hours, while the second step (solution to minimization problem \textbf{P3} with order of relaxation $2$) takes about $4.5$ hours.  The maximum amount of memory used by Matlab during the computation was about 1.9 GB. Based on the authors' experience, although sparsity is exploited, the identification problem considered in this example becomes computationally intractable  (in commercial workstations and using general purpose SDP solvers like SeDuMi) when models with more that 7 parameters are considered.
 \begin{figure}[h!]
\centerline{
\includegraphics[scale=1]{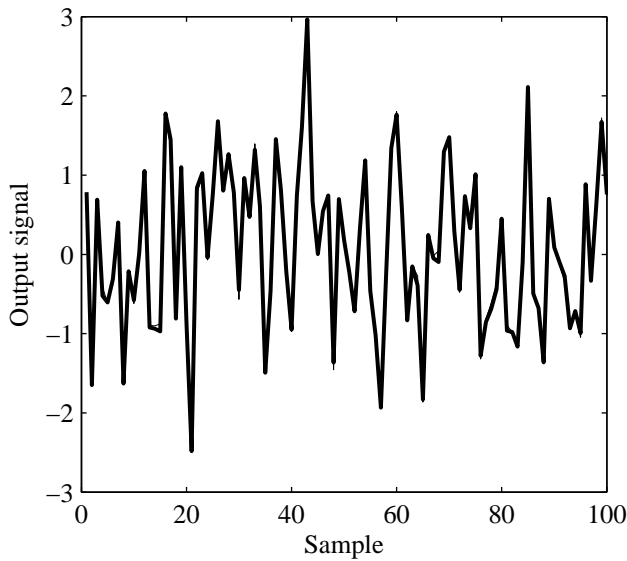}
}\caption{Noise-free output signal $\mathbf{w}(t)$ (thick line) and estimated output $\hat{\mathbf{y}}(t,\hat{\theta}^*)$ (thin line).}\label{Fese2:output}
\end{figure}
\\
\\
 \begin{figure}[h!]
\centerline{
\includegraphics[scale=1]{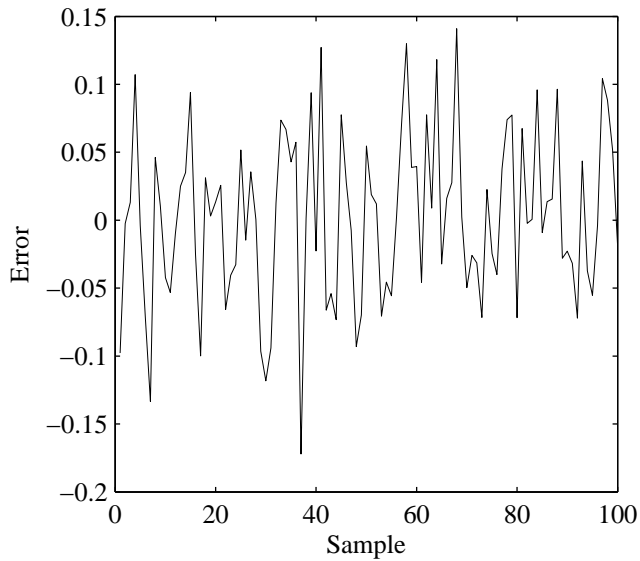}
}\caption{Estimate output error  $\mathbf{w}(t)-\hat{\mathbf{y}}(t,\hat{\theta}^*)$.}\label{Fese2:err}
\end{figure}

\section{CONCLUSIONS}
In this work we have presented a two-stage approach, based on suitable convex semidefinite relaxations, for approximating the global optimal solution to a general class of min-max constrained semialgebraic optimization problems arising in the framework of set-membership estimation theory. We have shown that the proposed methodology can be profitably applied to the problem of computing both \textit{conditional central} and \textit{robust projection} estimators in a nonlinear setting where the operator relating the data and the parameter to be estimated is assumed to be a generic multivariate polynomial function and the uncertainties affecting the data are assumed to belong to semialgebraic sets.
The key idea of the approach is to first compute a convergent polynomial approximation of the optimal value function of the inner maximization problem. Once such an approximation has been computed, the outer minimization problem reduces to a standard polynomial optimization problem solved by constructing a convergent hierarchy of semidefinite relaxations. Two simulation examples have been reported to show the effectiveness of the proposed approach. In particular, in the first example we have demonstrated that the presented two-stage algorithm provides good approximation of the global optimum of the considered min-max estimation problems, while in the second example we have shown that the proposed methodology can be applied to compute the solution to a challenging nonconvex constrained robust least square estimation problem.

\appendix

\subsection*{Proof of Lemma \ref{lemma-existence}.}
To prove that the semidefinite program (\ref{SOS_relax}) has an optimal solution, we
prove that Slater's condition holds for its dual, which is the semidefinite program:
\begin{equation}
\label{SOS_dual}
\begin{array}{ll}
\displaystyle\max_{\mathbf{z}}&L_\z(J(\theta,\alpha))\\
&\mathbf{M}_\tau(\z)\succeq0\\
&\M_{\tau-r_\mu}(d_\mu\,\z)\,\succeq\,0,\qquad \mu=1,\ldots,M\\
&\M_{\tau-1}(\phi_k\,\z)\,\succeq\,0,\qquad k=1,\ldots,\ell\\
&L_\z(\theta^\beta)\,=\,\gamma_\beta,\qquad\forall\beta\in\N^\ell_{2\tau}
\end{array}
\end{equation}
where
\begin{itemize}
\item $\z=(z_\kappa)$, $\kappa\in\N^{\ell+T}_{2\tau}$, is a sequence indexed in the canonical basis of monomials
$(\theta^\beta\alpha^\nu)$, of $\R[\theta,\alpha]_{2\tau}$ (the vector space of polynomials of degree at most $2\tau$).
\item $\M_\tau(\z)$ is the moment matrix of order $\tau$, associated with the sequence $\z$.
\item $\M_{\tau-r_\mu}(d_\mu\,\z)$ is the localizing matrix of order $\tau-r_\mu$, associated with the sequence $\z$
and the polynomial $d_\mu\in\R[\theta,\alpha]$ (and where $r_\mu\doteq\lceil({\rm deg}\,d_\mu)/2\rceil$).
\item $L_\z:\R[\theta,\alpha]\to\R$ is the so-called Riesz functional:
\[p\,\left(=\sum_{(\beta,\nu)\in\N^{\ell+T}}p_{\beta\nu}\,\theta^\beta\,\alpha^\nu\right)\quad \mapsto\,L_\z(p)\,=\,
\sum_{(\beta,\nu)\in\N^{\ell+T}}p_{\beta\nu}\,z_{\beta,\nu},\qquad p\in\R[\theta,\alpha].\]
\end{itemize}
For more details on moment and localizing matrices, and how they are used in polynomial optimization, the interested reader is referred to \cite{2010Lasbook}. Now let $O$ be an open set contained in $\mathcal{R}_\theta\times \mathcal{S}_{\alpha,\theta}$,
with projection $O_1$ on $\mathcal{R}_{\theta}$.
Let $\varphi$ be the Borel probability measure uniformly distributed on $\mathcal{R}_\theta$ with moments $(\gamma_\beta)$.
Let $\psi$ be the Borel probability measure on $\mathcal{R}_\theta\times \mathcal{S}_{\alpha,\theta}$ defined by:
\[\psi(A\times B)\,\doteq\,\int_A\Q(B\,\vert\,\theta)\,\varphi(d\theta),\qquad B\in\mathcal{B}(\mathcal{S}_{\alpha,\theta}),
A\in\mathcal{B}(\mathcal{R}_\theta),\]
where $\Q(\cdot\vert\cdot)$ is a stochastic kernel on $\mathcal{R}_\theta\times\mathcal{S}_{\alpha,\theta}$ such that
$\Q(\cdot\,\vert\,\theta)$ is the probability measure uniformly distributed on
$\mathcal{S}_{\alpha,\theta}$ if $\theta\in O_1$,
and $\Q(\cdot\vert\theta)$ is any probability measure on $\mathcal{S}_{\alpha,\theta}$ if $\theta\in\mathcal{R}_\theta\setminus O_1$.
Then let $\z$ be the sequence of moments associated with $\psi$, i.e.,
\[z_{\beta,\nu}\,\doteq\,\int \theta^\beta\,\alpha^\nu\,d\psi(\theta,\alpha),\qquad \forall\,(\beta,\nu)\in\N^{\ell+T}_{2\tau}.\]
Then $\M_{\tau}(\z)\succ0$, $\M_{\tau-r_\mu}(d_\mu\,\z)\succ0$, $\mu=1,\ldots,M$, and
$\M_{\tau-1}(\phi_k\,\z)\succ0$, $k=1,\ldots,\ell$ as well.
Indeed suppose for instance that $\M_\tau(\z)p=0$ for some vector $p\neq0$.
Let $\tilde{p}\in\R[\theta,\alpha]_{\tau}$ be the polynomial with coefficient vector $p$.
\[0\,=\,\langle p,\M_\tau(\z)\,p\rangle\,=\,\int \tilde{p}(\theta,\alpha)^2\,d\psi(\theta,\alpha)=0,\]
but this implies that $\tilde{p}$ vanishes on the whole open set $O$, in contradiction with $p\neq0$.
Similarly, let $\mu\in\{1,\ldots,M\}$ be arbitrary, and suppose that $\M_{\tau-r_\mu}(\z)p=0$ for some vector $p\neq0$,
and let $\tilde{p}\in\R[\theta,\alpha]_{\tau-r_\mu}$ be the polynomial with coefficient vector $p$.
\[0\,=\,\langle p,\M_{\tau-r_\mu}(\z)\,p\rangle\,=\,\int \tilde{p}(\theta,\alpha)^2\,d_\mu(\theta,\alpha)\,d\psi(\theta,\alpha)=0,\]
but this implies that $\tilde{p}$ vanishes on the whole open set $O$, in contradiction with $p\neq0$.
A similar argument shows that $\M_{\tau-1}(\phi_k\,\z)\succ0$ for every $k=1,\ldots,\ell$.
Moreover, from the definitions of $\psi$ and $\varphi$,
\[L_\z(\theta^\beta)\,=\,\int_{\mathcal{R}_\theta\times\mathcal{S}_{\alpha,\theta}}\theta^\beta\,d\psi(\theta,\alpha)\,=\,\int_{\mathcal{R}_\theta}\theta^\beta\,\varphi(d\theta)\,=\,\gamma_\beta,\qquad\beta\in\N^\ell_{2\tau},\]
and so $\z$ is admissible for (\ref{SOS_dual}). Therefore Slater's condition holds for (\ref{SOS_dual}), and by a well-known result of convex optimization, the dual of (\ref{SOS_dual}) (i.e. (\ref{SOS_relax})) has a optimal solution if its value is finite.

But the value of the primal semidefinite program (\ref{SOS_relax}) is bounded below by
$\int_{\mathcal{R}_\theta}J(\theta,\alpha)d\varphi(\theta)$. Moreover, there exists $M>0$ such that $M-J(\theta,\alpha)>0$ on
$\mathcal{R}_\theta\times\mathcal{S}_{\alpha,\theta}$. Therefore by Putinar's Positivstellensatz, there exists some integer $\tau_0$
such that
\[M - J(\theta,\alpha) = \sigma_0(\theta,\alpha)+\sum_{\mu=1}^M \sigma_\mu(\theta,\alpha)d_\mu(\alpha,\theta)+\sum_{k=1}^\ell\psi_k(\theta,\alpha)\,\phi_k(\theta)\]
where $\deg(\sigma_0) \leq 2\tau_0$,
$\deg(\sigma_\mu d_\mu) \leq 2\tau_0$, $\mu=1, \ldots, M$, and $\deg(\psi_k \phi_k)\leq 2\tau_0$,
$k=1, \ldots, \ell$.  Hence the optimal value of (\ref{SOS_relax}) is finite whenever $\tau\geq\tau_0$
and so (\ref{SOS_relax}) has an optimal solution. $\Box$

\bibliography{minmax}

\end{document}